\documentclass[reqno,final,12pt]{amsart}
\usepackage{amssymb,amsthm,amsmath,color,fullpage,url}
\usepackage{graphicx}
\usepackage{comment}
\usepackage[noabbrev,capitalise]{cleveref}

\newtheorem{thm}{Theorem}
\newtheorem{lem}[thm]{Lemma}

\newtheorem{cor}[thm]{Corollary}

\newcommand{\Po}{\mathcal{P}}
\newcommand{\Go}{\mathcal{G}}
\newcommand{\Ho}{\mathcal{H}}

\newcommand{\eps}{\varepsilon}
 
\newcommand{\mc}[1]{\mathcal{#1}} 
\newcommand{\bb}[1]{\mathbb{#1}}
\newcommand{\brm}[1]{\operatorname{#1}}

\usepackage[noabbrev,capitalise]{cleveref}
\newcommand{\ve}{\textup{\textsf{v}}}
\newcommand{\e}{\textup{\textsf{e}}}
\newcommand{\twr}{\textup{\textsf{twr}}}

\begin{document}
	
	\title{Testability and local certification of monotone
          properties in minor-closed classes}

\author[L.~Esperet]{Louis Esperet}
\address[L.~Esperet]{Univ.\ Grenoble Alpes, CNRS, Laboratoire G-SCOP,
  Grenoble, France}
\email{louis.esperet@grenoble-inp.fr}

\author[S.~Norin]{Sergey Norin}
\address[S.~Norin]{Department of Mathematics and Statistics, McGill
  University, Montreal, Canada}
\email{sergey.norin@mcgill.ca}

\thanks{L.\ Esperet is partially supported by the French ANR Projects
  GATO (ANR-16-CE40-0009-01), GrR (ANR-18-CE40-0032), TWIN-WIDTH
  (ANR-21-CE48-0014-01), and by LabEx
  PERSYVAL-lab (ANR-11-LABX-0025). S.\ Norin is supported by an NSERC Discovery grant.}

        \begin{abstract}
          The main problem in the area of graph property testing is to
          understand which graph properties are \emph{testable}, which
          means that with constantly many queries to any input graph $G$, a
          tester can decide with good probability whether $G$ satisfies the property, or is
          far from satisfying the property. 
          Testable properties are
          well understood in the dense model and in the bounded degree
          model, but little is known in sparse graph classes when
          graphs are allowed to have unbounded degree. This is the
          setting of the \emph{sparse model}.

          We prove that for any proper
          minor-closed class $\Go$, any monotone property (i.e., any property that is
          closed under taking subgraphs) is testable for graphs from
          $\Go$ in the sparse model. This extends a result of Czumaj
          and Sohler (FOCS'19), who proved it for
          monotone properties with finitely many forbidden subgraphs. Our result implies for instance
          that for any integers $k$ and $t$, $k$-colorability of $K_t$-minor free graphs is
          testable  in the sparse
          model.

          Elek recently proved that monotone properties of bounded
          degree graphs from minor-closed classes that are closed
          under disjoint union can be verified by an
    approximate proof labeling scheme in constant time. We show again
          that the assumption of bounded degree can be omitted in his result. 
\end{abstract}
        
	\maketitle

        \section{Introduction}

        \subsection{Property testing}

        We say that a graph $G$ is \emph{$\eps$-far from some property $\Po$}
        if one needs to modify at least $\eps |E(G)|$ of its adjacencies
        (replacing edges by non-edges and vice-versa) in order to
        obtain a graph satisfying $\Po$. A property is \emph{testable}
        if for any graph $G$, a tester can decide with good
        probability whether $G$ satisfies $\Po$ or is $\eps$-far from $\Po$, by only making a constant
        number of queries to a given representation of $G$ (i.e., the number of
        queries depends only on $\eps$ and $\Po$, but is
        independent of the input graph $G$). The tester has
        \emph{one-sided error}  if it always gives the correct answer when
        $G\in \Po$, and \emph{two-sided error} otherwise.

        \medskip

        In the \emph{dense graph model} \cite{GGR98}, there is a good
        understanding of which properties are testable with two-sided
        error \cite{AFNS09} and one-sided error
        \cite{AS08a,AS08b}. In \emph{the bounded degree model}
        \cite{GR02}, a sequence of papers \cite{BSS10,CSS09,HKNO09} culminated in a proof that
        every property is testable with two-sided error  within any \emph{hyperfinite} graph
        family (this includes for instance any proper  minor-closed
        class) \cite{NS13}. The bounded degree assumption is crucial
        for obtaining this result and it has since then been an
        important open problem to obtain testability results in the
        weaker \emph{sparse model}, which does not assume that the maximum
        degree is bounded \cite{CMOS19,CS19}. In this more general model there are two types of
        queries: given a vertex $v$, we can query the degree $d(v)$
        of $v$ in $G$; we can also query the $i$-th neighbor of $v$,
        for $1\le i \le d(v)$ (all these queries are assumed to take
        constant time). In this model, much less is known: it was
        proved that bipartiteness is testable within any minor-closed
        class in \cite{CMOS19}, while already in the bounded degree
        model many simple properties are not testable in general graph
        classes \cite{GR02}, so the restriction to a sparse
        structured class such as a proper minor-closed class is very
        natural in this context. The interested reader is referred to
        the book of Goldreich \cite{Gol17} for more results and
        references on property testing, and especially Chapter 10 in
        the book, which focuses on the general graph model.

        \medskip

        Instead of working in the sparse model as defined above, it
        will be enough to restrict ourselves to a single type of
        query: given a vertex $v$, we query a random neighbor of
        $v$, uniformly among the neighbors of
        $v$. Following \cite{CS19}, we say we make queries to the 
        \emph{random neighbor oracle}. Note that this type of queries
        can clearly be implemented is the sparse model, so this is a
        restriction of the model (see \cite{CS19} for a comparison
        between these two models, and a third one were we are allowed to query a
        constant number of 
        \emph{distinct} random neighbors of a given vertex).
        The following was recently proved by Czumaj and Sohler \cite{CS19}.

        \begin{thm}[\cite{CS19}]\label{thm:CS}
For every proper minor-closed class $\Go$, and any finite family
$\Ho$, the property of being $\Ho$-free for graphs from $\Go$ is
testable with one-sided error in the sparse model, where only queries
to the random neighbor oracle are allowed.
\end{thm}

Here we say that a graph is \emph{$H$-free} if it does not contain $H$ as a
subgraph, and \emph{$\Ho$-free} if it is $H$-free for every $H\in \Ho$.
Our main result is an extension of Theorem~\ref{thm:CS}  to any
\emph{monotone} property, that is any property closed under taking
subgraphs. 

\begin{thm}\label{thm:main}
For every proper minor-closed class $\Go$, and any monotone property
$\Po$, the property of satisfying $\Po$ for graphs from $\Go$ is
testable with one-sided error in the sparse model, where only queries
to the random neighbor oracle are allowed.
\end{thm}

Note that for any monotone property $\Po$ there is a (possibly
infinite) family of graphs $\Ho$ such that $\Po$ is precisely the property of
being $\Ho$-free. This family $\Ho$ can be simply defined as the class
of all the
graphs that do not satisfy $\Po$, or as the class of all the
graphs that do not satisfy $\Po$ and are minimal with this property (with respect to the subgraph
relation).  It
follows that Theorem~\ref{thm:main} is the natural generalization of
Theorem~\ref{thm:CS}, where we remove the assumption that $\Ho$ is
finite. This can be seen as an analogue of the situation in the dense
graph model: it was first proved that the property of being $H$-free
(or $\Ho$-free for finite $\Ho$) was testable in this model \cite{ADLRY94}, and then
only much later was this extended to all monotone classes by Alon and
Shapira~\cite{AS08a}. Note that many natural monotone properties, such as
being planar or $k$-colorable for some $k\ge 2$, do not have a finite set of
minimal forbidden subgraphs. So there is a fundamental gap between being $H$-free and
being $\Ho$-free for infinite $\Ho$.

\subsection{Local certification}
We now describe our second main result, which is obtained by
extending the methods used in the proof of \cref{thm:main}.
We start by introducing the setting of this result: The problem of \emph{local certification}.

In this part, all graphs are assumed to be connected. The vertices of any $n$-vertex
graph $G$ are assumed to be assigned distinct (but otherwise arbitrary)
identifiers $(\text{id}(v))_{v\in V(G)}$ from
$\{1,\ldots,\text{poly}(n)\}$. In the remainder of this section, all
graphs are implicitly labelled by these distinct identifiers (for
instance, whenever we talk about a subgraph $H$ of a graph $G$, we
implicitly refer to the corresponding labelled subgraph of $G$). We follow the
terminology introduced by G\"o\"os and Suomela~\cite{lcp}.

\subsubsection*{Proofs} A \emph{proof} for a graph $G$ is a function
$P:V(G)\to \{0,1\}^*$ ($G$ is considered as a labelled graph,
so the proof $P$ is allowed to depend on the identifiers of the
vertices of $G$). The binary words $P(v)$ are called \emph{certificates}. The \emph{size} of $P$ is the maximum size of a
certificate $P(v)$, for $v\in V(G)$.

\subsubsection*{Local verifiers} A \emph{verifier} $\mathcal{A}$ is a function that takes a
graph $G$, a proof $P$
for $G$, and a vertex $v\in V(G)$ in input, and outputs an element of
$\{0,1\}$. We say that $v$ \emph{accepts} the instance if
$\mathcal{A}(G,P,v)=1$ and that $v$ rejects the instance if
$\mathcal{A}(G,P,v)=0$.

Consider an integer $r\ge 0$, a graph $G$, a proof $P$ for $G$, and a
vertex $v\in V(G)$. Let $B_r(v)$ denote the set of vertices at
distance at most $r$ from $v$ in $G$. We denote by $G[v,r]$ the
subgraph of
$G$ induced by $B_r(v)$, and similarly we denote by $P[v,r]$ the restriction of $P$ to
$B_r(v)$.

A verifier $\mathcal{A}$ is \emph{local} if there is a constant $r\ge
0$, such that for any $v\in G$,
$\mathcal{A}(G,P,v)=\mathcal{A}(G[v,r],P[v,r],v)$. In other words, the
output of $v$ only depends on the ball of radius $r$ centered in $v$,
for any vertex $v$ of $G$. The constant $r$ is called the \emph{local
  horizon} of the verifier.

\subsubsection*{Proof labelling schemes}

For an integer $r\ge 0$, an \emph{$r$-round proof labelling scheme}
for a graph class $\Go$ is a prover-verifier pair
$(\mathcal{P},\mathcal{A})$, with the following properties.

\smallskip

\noindent {\bf $r$-round:} $\mathcal{A}$ is a local verifier with
local horizon at most $r$.

\smallskip

\noindent {\bf Completeness:} If $G\in \Go$, then
$P=\mathcal{P}(G)$ is a proof for $G$ such that for any vertex $v\in
V(G)$, $\mathcal{A}(G,P,v)=1$.

\smallskip

\noindent {\bf Soundness:}  If $G\not\in \Go$, then for every proof
$P'$ for $G$, there exists a vertex $v\in
V(G)$ such
that  $\mathcal{A}(G,P',v)=0$.

\medskip

In other words, upon looking at its ball of radius $r$ (labelled by
the identifiers and certificates), the local verifier of each vertex
of a graph $G\in \Go$ accepts the instance, while if $G\not\in
\Go$, for every possible choice of certificates, the local verifier of at least one vertex rejects the instance. 

\medskip

The \emph{complexity} of the labelling scheme is the maximum size of a
proof $P=\mathcal{P}(G)$ for an $n$-vertex graph $G\in\mathcal{F}$. If we say that the complexity is $O(f(n))$, for some
function $f$, the $O(\cdot)$ notation refers to $n\to \infty$. See
\cite{Feu21,lcp} for more details on proof labelling schemes and local
certification in general.

\medskip

It was proved in \cite{planar} that planar graphs have a 1-round proof
labelling scheme of complexity $O(\log n)$, and that this complexity is
optimal. The authors of \cite{planar} asked whether this can be
extended to any proper minor-closed class. This was indeed extended in \cite{genus} to graphs embeddable in a
fixed surface (see also \cite{EL} for a short proof), to graphs
avoiding some small minors in \cite{BFT}, and more generally to any
minor-closed class of bounded tree-width in \cite{tw} (in the last
result, the complexity is $O(\log^2 n)$ instead of $O(\log n)$ in the
other results mentioned here).

\medskip

For
$\eps>0$, define
an \emph{$r$-round $\eps$-approximate proof labelling scheme} for some
class $\Go$ exactly as in the definition of $r$-round proof labelling
scheme above, except that in the soundness part, the condition ``If
$G\not\in \Go$'' is replaced by ``If $G$ is $\eps$-far from
$\Go$'' \cite{CPP17}. A graph class $\Go$ is \emph{summable} if for any
$G_1,G_2\in \Go$, the disjoint union of $G_1$ and $G_2$ is also in
$\Go$. Elek recently proved the following result \cite{Ele20}. 

\begin{thm}[\cite{Ele20}]\label{thm:elek}
  For any $\eps>0$ and integer $D\ge 0$, and any monotone summable property $\Po$ of a proper
minor-closed class $\Go$, there are constants $r\ge 0$ and $K\ge 0$
such that  the class of graphs from $\Po$ with maximum degree at most
$D$  has an $r$-round $\eps$-approximate proof
labelling scheme of complexity at most $K$.
\end{thm}

A natural problem is whether the bounded degree assumption in Elek's
result can be omitted (Elek's proof crucially relies on this
assumption). We prove that the bounded degree assumption can indeed be omitted.

\begin{thm}\label{thm:lcertsum}
For any $\eps>0$ and any monotone summable property $\Po$ of a proper
minor-closed class $\Go$, there are constants $r\ge 0$ and $K\ge 0$
such that $\Po$ has an $r$-round $\eps$-approximate proof
labelling scheme of complexity at most $K$.
\end{thm}

We indeed prove a far-reaching generalization of this result
(whose statement was suggested by Elek to the authors), concerning graph classes with
bounded asymptotic dimension.

\subsubsection*{Asymptotic dimension}
Given a graph $G$ and an integer $r\ge 1$,
we denote by $G^r$ the graph obtained from $G$ by adding edges
between any pair of vertices at distance at most $r$ in $G$. The \emph{weak diameter} of a set $S$ of
  vertices of $G$ is the maximum distance in $G$ between two vertices
  of $S$.

  \smallskip
  
For an integer $d\ge 0$, a class of graphs $\Go$ has \emph{asymptotic
  dimension} at most $d$ if there is a
function $D:\mathbb{N}\to \mathbb{N}$ such that for any integer $r\ge
1$, any graph $G\in \Go$
has a $(d+1)$-coloring of its vertex set such that any monochromatic\footnote{A \emph{monochromatic component} in a colored graph $G$ is a
connected component of a subgraph of $G$ induced by one of the color classes.}
component of $G^r$ has
weak diameter at most $D(r)$ in $G$.

This notion was introduced by Gromov \cite{Gro93} in the more general context of
metric spaces. In the specific case of graphs, it was proved that
classes of bounded tree-width have asymptotic dimension at most 1, and
proper minor-closed classes have asymptotic dimension at most
2 \cite{asdim}. It was also proved that $d$-dimensional grids and families of graphs defined by the intersection
of certain objects (such as unit balls) in $\mathbb{R}^d$ have asymptotic
dimension $d$ \cite{asdim}. On the other hand, it is known that any class of bounded
degree expanders has infinite asymptotic dimension (see~\cite{Hum17}).

\smallskip

We will prove  the following  generalization of Theorem \ref{thm:lcertsum}.

\begin{thm}\label{thm:lcertsumasdim}
For any $\eps>0$ and any monotone summable property $\Po$ of a class
$\Go$ of bounded asymptotic dimension, there are constants $r\ge 0$ and $K\ge 0$
such that $\Po$ has an $r$-round $\eps$-approximate proof
labelling scheme of complexity at most $K$.
\end{thm}

Note that a monotone
property  $\Po$ is summable if and only if  all minimal forbidden subgraphs for $\mc{P}$ are connected. This includes for instance minor-closed classes whose
minimal forbidden minors are connected, such as planar graphs, $K_t$-minor
free graphs for any $t\ge 2$, graphs of bounded tree-width, graphs of
bounded tree-depth, and graphs of bounded Colin de Verdi\`ere
parameter. 



\medskip

Natural examples of non summable properties include
toroidal graphs (or more generally graphs embeddable on any fixed surface other than
the sphere).
For monotone properties that are not necessarily summable, we prove
the following. 

\begin{thm}\label{thm:lcert}
For any $\eps>0$ and any monotone property $\Po$ of a proper
minor-closed class $\Go$, $\Po$ has a 1-round $\eps$-approximate proof
labelling scheme of complexity $O(\log n)$.
\end{thm}

While the complexity of the scheme guaranteed by \cref{thm:lcert} is not constant as in Elek's result
\cite{Ele20} and Theorem~\ref{thm:lcertsum}, we do not require any bounded degree assumption (as in Theorem~\ref{thm:lcertsum}), and
a local horizon of $1$ is sufficient. 
More importantly, the fact that $\Po$ is not necessarily summable
requires a completely different set of techniques, much closer from
the tools used to prove Theorem \ref{thm:main}.
\cref{thm:lcert}  can be
thought of as an approximate answer to the question of \cite{planar} on
the local certification of minor-closed classes.

\medskip

\subsubsection*{Organization of the paper} We start with some preliminary results in Section~\ref{sec:prel}. In
Section~\ref{sec:ob}, we prove the main technical contribution of this
paper, a result showing that if a graph from some minor-closed class is far from a monotone property
$\Po$, then it contains linearly many edge-disjoint subgraphs of
bounded size that are not in $\Po$. In Section~\ref{sec:proof} we deduce Theorem~\ref{thm:main} from this result and
Theorem~\ref{thm:CS}. Theorems~\ref{thm:lcertsum} and~\ref{thm:lcert} are proved in
Section~\ref{sec:lcert}. We conclude in Section~\ref{sec:cl} with some
remarks.

\section{Preliminaries}\label{sec:prel}

\subsubsection*{Minor-closed classes} We denote the number of vertices of a graph $G$ by $\ve(G)$, and its
number of edges by $\e(G)$. A class of graphs $\Go$ is \emph{minor-closed}
if any minor of a graph from $\Go$ is also in $\Go$. A class is \emph{proper}
if it does not contain all graphs.  The following was proved by Mader
\cite{Mad67}.

\begin{thm}[\cite{Mad67}]\label{thm:Mader}
  For any
proper minor-closed class $\Go$, there is a constant $C$ such that for
any graph $G\in \Go$, $\e(G)\le C \,\ve(G)$.
\end{thm}

\subsubsection*{Tree-depth} Given a rooted tree $T$, the
\emph{closure} of $T$ is the graph obtained from $T$ by adding edges
between each vertex and its ancestors in the tree. The \emph{height} of
a rooted tree is the maximum number of vertices on a root-to-leaf path
in the tree. The
\emph{tree-depth} of a connected graph $G$ is the maximum height of a rooted tree $T$ such
that $G$ is a subgraph of the closure of $T$, and the tree-depth of a
graph $G$, denoted by $\brm{td}(G)$, is the maximum tree-depth of its connected components
(equivalently, it is equal to the maximum height of a rooted
\emph{forest} $F$ such that $G$ is a subgraph of the closure of $F$).

\medskip

The following was implicitly proved by Dvo\v r\'ak and Sereni
\cite{DS20} (in the proof below the actual definition of tree-width is
not needed, so we omit it). 

\begin{thm}[\cite{DS20}]\label{t:DS}
	For every proper minor-closed class $\mc{G}$ and every $\delta
        > 0$ there exists $d=d_{\ref{t:DS}}(\mc{G},\delta) \in \bb{N}$
        and $s=s_{\ref{t:DS}}(\mc{G},\delta) \in \bb{N}$ satisfying the following. For every $G \in \mc{G}$ there exist $X_1,X_2, \ldots, X_s \subseteq V(G)$ such that  
	\begin{itemize}
		\item for any $1\le i \le s$, $\brm{td}(G[X_i]) \leq d$, and
		\item every $v \in V(G)$ belongs to at least $(1 -
                  \delta)s$ of the sets $X_i$.
	\end{itemize} 
      \end{thm}

      \begin{proof} Let $t=\lceil
      	\tfrac2\delta\rceil$.
It was proved in \cite{DDOSRS} that there is
a constant $k=k(t,\mc{G})$ such that any graph $G\in \Go$ has a partition of its
vertex set into $t$ classes $Y_1,\ldots,Y_{t}$, such that the union of
any $t-1$
classes $Y_ i$ induces a graph
of tree-width at most $k$. In particular, if we define
$Z_i:=V(G)\setminus Y_i$ for any $1\le i \le t$, then each graph
$G[Z_i]$ has tree-width at most $k$ and each vertex $v$ lies in $t-1=(1-\tfrac1t)t$ sets $Z_i$.
Dvo\v r\'ak and Sereni \cite[Theorem 31]{DS20} proved\footnote{The property that $s$ is bounded
  independently of $G$ does not appear explicitly in the statement of their
  theorem, but readily follows from their proof. This will only be needed  in Section~\ref{sec:lcert}.}  that for every
integer $k$ and real $\delta>0$, there are integers $r=r(k,\delta)$
and $d=d(k,\delta)$
such that for any graph $H$ of tree-width at most $k$, $H$ has a cover of its
vertex set by $r$ sets $X_1,\ldots,X_r$, such that each $H[X_i]$ has
tree-depth at most $d$ and each vertex lies in at least $(1-\tfrac\delta{2})r$
sets $X_i$. Applying this result to $H=G[Z_i]$ for any $1\le i \le t$,
we obtain $rt$ sets $X_1',\ldots,X_{rt}'$ of vertices of $G$, such
that the subgraph  $G[X_i']$ induced by each of them has tree-depth at most $d$ and each vertex of $G$
lies in at least $(1-\tfrac{\delta}2)r\cdot (1-\tfrac1t)t\ge (1-\delta)rt$
sets $X_i'$. Thus $d$ and $s=rt$ satisfy the conditions of the theorem.
      \end{proof}
      
We deduce the following useful result.

\begin{cor}\label{c:td}
	For every proper minor-closed class $\mc{G}$ and every $\eps > 0$ there exists $d=d_{\ref{c:td}}(\mc{G},\eps) \in \bb{N}$ satisfying the following. For every $G \in \mc{G}$ there exist $F \subseteq E(G)$ such that $|F| \leq \eps \,\e(G)$ and $\brm{td}(G \setminus F) \leq d$. 
\end{cor}
\begin{proof}
        Let $\delta= \frac{\eps}{2}$. We show that $d=d_{\ref{t:DS}}(\mc{G},\delta)$ satisfies the corollary.
	Indeed, for  $G \in \mc{G}$ let  $X_1,X_2, \ldots, X_s
        \subseteq V(G)$ be as in \cref{t:DS}. Let $F_i=E(G)\setminus
        E(G[X_i])$ for $i \in [s]$, then $\brm{td}(G \setminus F_i)
        \leq d$. Moreover, every edge belongs to at most $2 \delta s$
        sets $F_i$, so $$\frac{1}{s} \sum_{i=1}^{s}|F_i|\leq \frac1s
        \cdot 2 \delta s\cdot  \e(G) = \eps\, \e(G).$$ 
	Thus, by averaging, $|F_i| \leq \eps \,\e(G)$ for some $i$, and $F=F_i$ satisfies the corollary.	
      \end{proof}

     Note that the conclusion of Corollary~\ref{c:td} can be shown to
     hold in greater generality than in the context of minor-closed
     classes. For instance, any class in which all graphs can be made
     of bounded \emph{tree-width} by removing an arbitrarily small fraction of
     edges also have this property (see~\cite{DS20}). This includes all graphs of
     bounded \emph{layered tree-width} (see \cite{DMW17,Sha13}). Typical non minor-closed examples of such classes
     are families of graphs that can be embedded on a fixed surface,
     with a bounded number of crossings per edge~\cite{DMW19}. However, since
     the proof of Theorem~\ref{thm:CS} itself strongly relies on
     edge-contractions (and thus on the graph class $\Go$ being
     minor-closed), Theorem~\ref{thm:main} does not seem to be easily extendable
     beyond minor-closed classes.

        \section{Bounded size obstructions}\label{sec:ob}

\subsection{General properties}
        
A graph property $\Po$ is a graph class that is closed under
isomorphism. It will be convenient to write that $G\in \Po$ instead of
``$G$ satisfies $\Po$'' in the remainder of the paper. 
A graph $H$ is
\emph{minimally not in $\Po$} if $H\not\in \Po$ and any proper
subgraph of $H$ is in $\Po$.

\medskip

We will use the following result of Ne\v set\v ril and Ossona de
Mendez (Lemma 6.13 in \cite{NO12}) \footnote{The version we use here
  only needs $Q$ to be a singleton in the statement of Lemma 6.13 in
  \cite{NO12}.}.

\begin{lem}[\cite{NO12}]\label{lem:NO} For every integer $d\ge 1$ and every property $\Po$, there exists $N=N_{\ref{lem:NO}}(d,\Po)$ such
  that if $H$ is minimally not in $\Po$ and $\brm{td}(H) \leq d$ then $\ve(H) \leq N$.
\end{lem}

\subsection{Colorability}

The conclusion of Lemma~\ref{lem:NO} is quite strong but it does not give explicit bounds
on $N_{\ref{lem:NO}}(d,\Po)$. For completeness, we give such an explicit bound when $\Po$ is the property of being
$k$-colorable. The
specific question of whether 3-colorability of planar graphs was testable in the sparse model was raised by Christian Sohler at the Workshop on Local
Algorithms (WOLA) in 2021. A positive answer to this question directly
follows from Theorem~\ref{thm:main}, but the lemma below allows us to
give an explicit
bound on the query complexity of testing $k$-colorability in
minor-closed classes (see Section~\ref{sec:cl}).

\medskip

Given a graph $H$ and two vertex subsets $A,B\subseteq V(H)$, we say that
$(A,B)$ is a \emph{proper separation} of $H$ if $A\cup B=V(H)$,
$A\setminus B$ and  $B\setminus A$ are both non-empty, and there are
no edges between $A\setminus B$ and $B\setminus A$ in $H$.
We say that a graph $H$ is \emph{split} if there exists a proper separation
$(A,B)$ of $H$ and an isomorphism
$\phi: A \to B$ between $H[A]$ and $H[B]$ such that $\phi(v)=v$ for
every $v \in A \cap B$ (see Figure \ref{fig:split} for an example). Equivalently, a split graph can be obtained
  by taking two copies of some smaller graph and, for a proper subset
  of vertices, identifying the two copies of the vertex subset with
  each other.

  We say that a connected graph $H$ is \emph{unsplit} if it is not
split. Note that minimally non-$k$-colorable graphs are unsplit.

\begin{figure}[htb]
 \centering
 \includegraphics[scale=1.2]{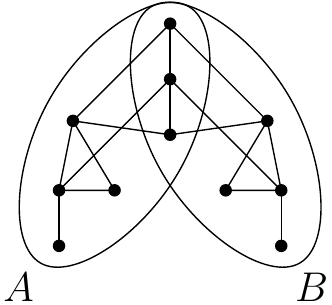}
 \caption{A split graph $H$ and the corresponding proper separation
   $(A,B)$ of $H$. The isomorphism $\phi:A\to B$ is the reflection symmetry
   with respect to the vertical axis.}
 \label{fig:split}
\end{figure}

\smallskip

The
\emph{tower function} is defined as $\twr(0)=0$ and
$\twr(i+1)=2^{\twr(i)}$ for any integer $i\ge 0$.

\begin{lem}\label{l:unsplit}
	For every integer $d\ge 1$, there exists $N=N_{\ref{l:unsplit}}(d)=\twr(O(d))$ such that if $H$ is an unsplit graph with $\brm{td}(H) \leq d$ then $\ve(H) \leq N$.
\end{lem}

\begin{proof}
Define $k_d:=1$ and for any $0\le i\le d-1$, let
  $k_i:=(2^{{k_{i+1}\choose 2}+d \cdot k_{i+1}})k_{i+1}+1$. 

Choose a rooted tree $T$ of height at most $d+1$ rooted at some vertex $r$, such that $H$ is a subgraph of the closure of $T$.
 For each $v \in V(H)$,	let $T_v$ be the subtree of $T$ rooted at $v$ (consisting of $v$ and all its descendants), and let $H_v$ be the subgraph of $H$ induced by $V(T_v)$.
 Define the \emph{level} $\ell(v):= \brm{dist}_T(r,v) \in[0,d]$.
 
 We prove by induction on $d-i$ that $\ve(H_v) \leq k_i$ for every $v$
 with $\ell(v)=i$. The base case $i=d$ trivially holds, as $k_d=1$.
 
 For the induction step, let $u_1,\ldots,u_m$ be all the children of a
 vertex $v$ with $\ell(v)=i$, and let $A_v$ be the set consisting of
 $v$ and its ancestors. Then $|A_v|=i+1$. For each $j \in [m]$, let
 $H^{+}_j = H[A_v \cup V(H_{u_j})]$.
Note that there are at most $2^{{k_{i+1}\choose 2}+d \cdot
  k_{i+1}}$ distinct (labelled) graphs $G$ on at most $|A_v|+k_{i+1}\le d+k_{i+1}$
vertices such that the subgraph of $G$ induced by  the first $|A_v|$
vertices is
isomorphic to $H[A_v]$.  
 As $\ve(H_{u_j}) \leq k_{i+1}$, if
 $m>2^{{k_{i+1}\choose 2}+d \cdot k_{i+1}}$  there exist $j \neq j'$
 and an isomorphism $\phi: V(H^{+}_j) \to V(H^{+}_{j'})$ such that
 $\phi(w)=w$ for every $w \in A_v$. Such an isomorphism would imply
 that $H$ is split, and so $m\le 2^{{k_{i+1}\choose 2}+d \cdot
   k_{i+1}}$. It follows that $\ve(H_v) \leq mk_{i+1}+1\le (2^{{k_{i+1}\choose 2}+d
   \cdot k_{i+1}})k_{i+1}+1=k_i$, as desired.

 By taking $N:=k_0$ we obtain that $\ve(H)\le N$. It can be
  checked from the definition of $(k_i)_{0\le i \le d}$ that $N=k_0$ is at most a tower function of $O(d)$.
\end{proof}

\subsection{A linear Erd\H os-Pos\'a  property}

We use Corollary \ref{c:td} and Lemma \ref{lem:NO} to deduce the
following result, which is the main technical contribution of this paper.

\begin{thm}\label{thm:LEP}
	For  every proper minor-closed class $\mc{G}$, every $\eps >
        0$, and every property $\Po$, there exists $\delta
        > 0$ and an integer $N$ such that for every $G \in \mc{G}$ either
	\begin{itemize}
		\item there exists $F \subseteq E(G)$ with $|F| \leq
                  \eps\, \e(G)$ such that $G \setminus F$ is in $\Po$, or
		\item there exist edge-disjoint subgraphs
                  $G_1,\ldots, G_m$ of $G$ that are not in $\Po$, such
                  that $m \geq \delta \,\e(G)$ and for every $1\le i
                  \le m$, $\ve(G_i)\le N$.
 	\end{itemize}	
\end{thm}

\begin{proof}
By Theorem~\ref{thm:Mader}, there exists  $C$  such that $\e(G) \leq C
\ve(G)$ for every $G \in \mc{G}$.
Let
$d:=d_{\ref{c:td}}(\mc{G},\eps/2)$ and let $N:=N_{\ref{lem:NO}}(d,\Po)$. We show that $\delta:=\frac{\eps}{2NC}$ satisfies the theorem.

Let $G_1,\ldots, G_m$ be a maximal collection of edge-disjoint
subgraphs of $G$ that are not in $\Po$, and such that $\ve(G_i) \leq N$. If $m \geq \delta \e(G)$ the theorem holds, so we assume that $m < \delta \e(G)$.

Let $F' = \bigcup_{i=1}^m E(G_i)$. Then $$|F'| \leq C \sum_{i=1}^m \ve(G_i) \leq C m N < cN  \delta \e(G) \leq \frac{\eps}{2} \,\e(G).$$
Let $G' = G \setminus F'$.  By the choice of $d$, it follows from
Corollary~\ref{c:td} that  there exists $F'' \subseteq E(G')$ such that $|F''| \leq \frac{\eps}{2} \,\e(G)$ and $\brm{td}(G' \setminus F'') \leq d$. 

Let $G''=G' \setminus F''$. Suppose first that $G''$ is not in $\Po$,
and let $H$ be a minimal subgraph of $G''$ that is not in
$\Po$. As $H$ is minimally not in $\Po$, it follows from
Lemma~\ref{lem:NO} that $\ve(H) \leq N$. Thus adding $H$ to the collection $G_1,\ldots, G_m$ contradicts its maximality.

	It follows that $G''$ is in $\Po$, but $G'' = G \setminus F$, where $F = F' \cup F''$ and $|F| \leq  \eps\, \e(G)$, and so the theorem holds.
\end{proof}

\section{Property testing in the sparse model}\label{sec:proof}

\subsubsection*{The model} As alluded to in the introduction, we work
in the sparse model, only using queries to the random neighbor
oracle. That is, given an input graph $G$, the tester only does a
constant number of queries to the input, all of the following type:
given a vertex $v$, return a random neighbor of $v$ (uniformly at
random among all the neighbors of $v$ in $G$). The vertex $v$ itself
can be taken to be a random vertex of $G$, but does not need to. The
computation of a random vertex of $G$ and a random neighbor of a given
vertex of $G$ are assumed to take constant time in this model.

\medskip

We are now ready to prove Theorem~\ref{thm:main}.

\medskip

\noindent {\it Proof of Theorem \ref{thm:main}.}
Let $\Go$ be a
proper-minor class and $\Po$ be a monotone property. Let $\eps$ be
given. Let $\delta>0$ and $N$ be obtained by applying
Theorem~\ref{thm:LEP} to $\Go$, $\Po$ and $\eps$, and let $\Ho$ be the (finite) set of all graphs of at
most $N$ vertices that are not in $\Po$. We now run the tester of
Theorem~\ref{thm:CS} for testing whether a graph $G\in \Go$ is
$\Ho$-free or $\delta$-far from being $\Ho$-free.

Assume first that $G\in \Po$. If $G$ contains a graph $H\in
\Ho$ as a subgraph, then since $\Po$ is monotone, we have $H\in \Po$, which is
a contradiction. Hence, $G$ is $\Ho$-free, and it follows that  the one-sided tester of Theorem~\ref{thm:CS} accepts $G$ with probability 1. Assume now that $G$ is $\eps$-far from $\Po$. By
Theorem~\ref{thm:LEP}, there exist at least $\delta \e(G)$
edge-disjoint subgraphs of $G$ that are all in $\Ho$, and
thus one needs to remove at least one edge in each of these $\delta \e(G)$
edge-disjoint subgraphs to obtain an
$\Ho$-free graph. As $\Po$ is monotone, $G$ is $\delta$-far from being
$\Ho$-free, and it follows that the tester of 
Theorem~\ref{thm:CS} rejects $G$ with probability at
least $\tfrac23$, as desired.
This concludes the proof of Theorem~\ref{thm:main}.\hfill $\Box$

\section{Local Certification}\label{sec:lcert}

We recall that in this part, all graphs are assumed to be
connected.

\subsection{Summable properties}

Before we prove Theorem \ref{thm:lcertsumasdim}, we will need the
following consequence of a result of Brodskiy, Dydak, Levin and Mitra~\cite{BDLM}
(obtained by taking $r=2$ in Theorem 2.4 in their paper). This can be
seen as an analogue of Theorem~\ref{t:DS} where tree-depth is replaced
by the weaker notion of weak diameter, while proper minor-closed
classes are replaced by the more general classes of bounded asymptotic dimension.

\begin{thm}[\cite{BDLM}]\label{t:asdim}
Let $\Go$ be a class of graphs of bounded asymptotic dimension and let
$\delta>0$ be a real number. Then there exist two constants
$D=D_{\ref{t:asdim}}({\Go,\delta})\in \bb{N}$ and $s=s_{\ref{t:asdim}}({\Go,\delta})\in \bb{N}$ satisfying the following. For every $G \in \mc{G}$ there exist $X_1,X_2, \ldots, X_s \subseteq V(G)$ such that  
	\begin{itemize}
		\item for any $1\le i \le s$, each connected component of
                  $G[X_i]$ has weak diameter at most $D$ in $G$, and
		\item every $v \in V(G)$ belongs to at least $(1 -
                  \delta)s$ of the sets $X_i$.
                \end{itemize} 
\end{thm}

 It can be noted that if a subset $S$ of
vertices of a graph $G$ is such that $G[S]$ has
bounded tree-depth, then $G[S]$ has bounded diameter~\cite{NO12},
and thus $S$ has bounded weak diameter in $G$. It follows that in the
special case of
proper minor-closed classes,
Theorem~\ref{t:DS} implies Theorem \ref{t:asdim} in a strong form.

\medskip

\noindent{\it Proof of Theorem~\ref{thm:lcertsumasdim}.}
Fix any real number $\eps>0$ and an integer $d\ge 0$. Let $\Go$ be a
class of bounded asymptotic dimension, and let $\Po$ be a monotone summable property of $\Go$. Let
$\delta=\tfrac{\eps}2$. By
Theorem~\ref{t:asdim}, there exist two constants
$D=D_{\ref{t:asdim}}({\Go,\delta})\in \bb{N}$ and $s=s_{\ref{t:asdim}}({\Go,\delta})\in \bb{N}$ satisfying the following. For every $G \in \mc{G}$ there exist $X_1,X_2, \ldots, X_s \subseteq V(G)$ such that  
	\begin{itemize}
		 \item for any $1\le i \le s$, each component of
                  $G[X_i]$ has weak diameter at most $D$ in $G$, and
		\item every $v \in V(G)$ belongs to at least $(1 -
                  \delta)s$ of the sets $X_i$.
                \end{itemize}

For any $v\in V(G)$, we define the proof $P(v)$ as (a binary
representation of) the set of indices $I(v)\subseteq\{1,2,\ldots,s\}$ such that $v\in X_i$. This proof has constant size
(depending only of $\Po$ and $\eps$).

For every vertex $v$, the local verifier $\mathcal{A}(G,P,v)$ first checks
that $I(v)$ contains at least $(1-\delta)s$ integers from
$\{1,2,\ldots,s\}$. If this is not the case, then $v$ rejects the
instance.
In the remainder, we call a \emph{monochromatic component of
  color $i$} a maximal connected subset of vertices $v$
of $G$ such that $i\in I(v)$. We omit the color if it is irrelevant in
the discussion. Note that each vertex $v$ belongs to $|I(v)|$
monochromatic components. For each vertex $v$, $\mathcal{A}(G,P,v)$ checks that the subgraph of $G$ induced by
the vertices $u\in B_{r}(v)$ is in $\Po$, for $r=2D+1$, and that all monochromatic
components of $G$ containing $v$ have weak diameter at most $D$ (this can
be clearly done as $v$ has access to the subgraph of $G$ induced by  its ball of
radius $r=2D+1$). If this is the case, then $v$ accepts
the instance, and otherwise $v$ rejects the instance.

\medskip

It follows from
the definition of our scheme and the monotonicity of $\Po$ that for any $G\in \Po$, the local
verifier $\mathcal{A}(G,P,v)$ of each vertex $v$ of
$G$ accepts the instance.
Consider now a graph $G$ and a proof $P'$ such that for each vertex
$v$, $\mathcal{A}(G,P',v)=1$. The proof $P'$
assigns a subset $I(v)$ of indices of $\{1,\ldots,s\}$ to each vertex
$v$ of $G$, such that $|I(v)|\ge (1-\delta)s$. For each $1\le i \le s$, let
$X_i$ be the subset of vertices $v$ of $G$ such that $i\in I(v)$. Let $C$ be
a connected component of some  $G[X_i]$ (that is, $C$ is a
monochromatic component of color $i$), for some $1\le i \le
s$. Since all vertices of $C$ accept the instance, $C$ has weak
diameter at most $D$ in $G$. It follows that $C$ is contained in some ball of
radius $D$ in $G$, and thus (since $\Po$ is monotone, and each ball of
radius $r=2D+1\ge D$ induces a graph of $\Po$), $C$ lies in $\Po$. As $\Po$ is
summable, $G[X_i]$ also lies in $\Po$.

It follows from the proof of  Corollary
                \ref{c:td}, that if we set $F_i=E(G)\setminus
                E(G[X_i])$ for any $1\le i \le s$, then the property
                that every vertex $v \in V(G)$ belongs to at least $(1 -
                  \delta)s$ of sets $X_i$ implies that there is an
                  index $1\le i \le s$ such that $|F_i|\le \eps\,
                  \e(G)$. By the paragraph above $G\setminus F_i$
                  satisfies $\Po$, and thus $G$ is $\eps$-close from
                  $\Po$ (we say that a graph is \emph{$\eps$-close} from
                  $\Po$ if it is not $\eps$-far from $\Po$).  In the contrapositive, we have proved that if $G$ is $\eps$-far from
  $\Po$, then there is at least one vertex $v$ such that
  $\mathcal{A}(G,P',v)=0$, as desired.
  \hfill$\Box$

  \medskip

Using the fact that proper minor-closed classes have asymptotic
dimension at most 2~\cite{asdim}, we immediately obtain 
Theorem~\ref{thm:lcertsum} as a corollary. Note that for the same purpose we could also use an  earlier (and
simpler)  result of
Ostrovskii and Rosenthal~\cite{OR15}, who proved that for every
integer $t$, the class of $K_t$-minor free graphs has asymptotic
dimension at most $4^t$. We could also use Theorem~\ref{t:DS} without
any reference to asymptotic dimension,
as Theorem~\ref{t:DS} implies Theorem \ref{t:asdim} for proper minor-closed
classes (see the discussion before the proof of Theorem~\ref{thm:lcertsumasdim}).

\subsection{Non necessarily summable properties}

We now consider proof labelling schemes of complexity $O(\log n)$,
rather than $O(1)$. To prove Theorem~\ref{thm:lcert}, we will need the following recent result of Bousquet, Feuilloley and
Pierron~\cite{BFT2}.

\begin{thm}[\cite{BFT2}]\label{thm:bft}
For every integer $d\ge 1$ and every first-order sentence $\varphi$, the
class of graphs of tree-depth at most $d$ satisfying $\varphi$ has a
1-round proof labelling scheme of complexity $O(\log n)$.
\end{thm}

Although we will not need it, it is worth noting that the complexity in their result is of order
$O(d\log n+ f(d,\varphi))$. Observe also that checking whether a graph is $\Ho$-free for
some fixed finite family $\Ho$ can be expressed by a first-order
formula. In particular, it follows directly from Lemma~\ref{lem:NO}
that checking whether $\brm{td}(G)\le d$ can be expressed by a
first-order formula (and thus certified with local horizon 1 with labels of size
$O(\log n)$ per vertex).

\medskip

The final ingredient that we will need is the ability to certify a
rooted spanning tree, together with the children/parent relationship
in this tree, with certificates of $O(\log n)$ bits
per vertex (see~\cite{AKY97,APV91,IL94} for the origins of this
classical scheme). The prover gives  the identifier $\mathrm{id}(r)$ of the
root $r$ of $T$ to each vertex $v$ of $G$, as well as $d_T(v,r)$, its
distance to $r$ in $T$, and each vertex $v$ distinct from the root is
also given the identifier of its parent $p(v)$ in $T$. The local
verifier at $v$ starts by checking  that $v$ agrees with all its
neighbors in $G$ on the identity of the root $r$ of $T$. If so, if
$v\ne r$, $v$ checks that $d_T(v,r)=d_T(p(v),r)+1$. It can be checked
that all vertices accept the instance if and only $T$ is a rooted
spanning tree of $G$. Moreover, once the rooted spanning tree $T$ has been certified, each vertex of
$G$ knows its parent and children (if any) in $T$.

\medskip

We are now ready to prove Theorem~\ref{thm:lcert}.

\medskip

\noindent{\it Proof of Theorem~\ref{thm:lcert}.}
The beginning of the proof proceeds exactly as in the proof of Theorem~\ref{thm:lcertsum}.
Fix any real number $\eps>0$. Let $\Go$ be a proper minor-closed
class, and let $\Po$ be a monotone (not necessarily summable) property of $\Go$. Let
$\delta=\tfrac{\eps}2$. By
Theorem~\ref{t:DS}, there exist $d=d_{\ref{t:DS}}(\mc{G},\delta) \in \bb{N}$
        and $s=s_{\ref{t:DS}}(\mc{G},\delta) \in \bb{N}$ satisfying the following. For every $G \in \mc{G}$ there exist $X_1,X_2, \ldots, X_s \subseteq V(G)$ such that  
	\begin{itemize}
		\item for any $1\le i \le s$, $\brm{td}(G[X_i]) \leq d$, and
		\item every $v \in V(G)$ belongs to at least $(1 -
                  \delta)s$ of the sets $X_i$.
                \end{itemize}
                By Lemma~\ref{lem:NO}, there exists a
constant $N=N_{\ref{lem:NO}}(d,\Po)$ such
  that if $H$ is minimally not in $\Po$ and $\brm{td}(H) \leq d$ then
  $\ve(H) \leq N$. Let $\Ho$ be the (finite) set of all graphs of at
  most $N$ vertices that are not in $\Po$.

  \medskip

For any $v\in V(G)$, the proof $P(v)$ contains (a binary
representation of) the set of indices
$I(v)\subseteq \{1,\ldots,s\}$ such that $v\in X_i$. This part of the proof has constant size
(depending only of $\Po$ and $\eps$). As in the proof of
Theorem~\ref{thm:lcertsum}, the local verifier at each vertex $v$
checks that $|I(v)|\ge (1 -
                  \delta)s$, and rejects the instance if this does not
                  hold.

                  \medskip

For each $1\le i \le s$, we do the following. In each
connected component $C$ of $G[X_i]$, we consider a rooted spanning tree
$T_C$ of $C$, with root $r_C$, and certify it using certificates of
$O(\log n)$ bits per vertex. It follows from Theorem~\ref{thm:bft}
that any first-order property of $G[C]$ can be 
certified with certificates of size $O(\log n)$ bits per vertex (as
all the components $C$ are vertex-disjoint, combining all these
certificates and schemes still results in a scheme with labels of
$O(\log n)$ bits per vertex). In particular we can certify that $\brm{td}(G[C])\le d$ (this is a first-order
property). Let $\Ho'$ be the class of all (non-empty) graphs obtained from a graph
$H\in \Ho$ by deleting an arbitrary subset of connected components of
$H$ (note that if all the graphs of $\Ho$ are connected,
$\Ho=\Ho'$). Observe that all the graphs of $\Ho'$ have size at most
$N$ (which is a constant independent of the size of $G$).
Then, for any $H'\in \Ho'$,  we certify that $G[C]$ is $H'$-free or
contains a copy of $H'$ using Theorem~\ref{thm:bft},
and store this information at the vertex $r_C$ in a constant-size binary array $b(r_C)$, whose entries
are indexed by all the graphs of $\Ho'$ (where the entry of $b(r_C)$
corresponding to some $H'\in \Ho'$ is equal to 1 if and only if $C$
contains a copy of $H'$ as a subgraph).

\smallskip

It remains to aggregate this information along some rooted spanning tree
$T$ of $G$ (which can itself be certified with certificates of $O(\log
n)$ bits per vertex). We do this as follows, for every $1\le i \le s$. For a vertex $v$
of the rooted tree $T$, the subtree of $T$ rooted in $v$ is denoted by
$T_v$. For each vertex
$v$ of $G$, let $\mathcal{C}_v$ be  the set of  components $C$ of
$G[X_i]$ such
that $r_C$ lies in $T_v$. Then the proof $P(v)$ contains a
binary array $c(v)$, whose entries are indexed by the graphs $H'$ of $\Ho'$. The array
$c(v)$ is defined as follows: for any $H'\in\Ho'$, the entry of $c(v)$
corresponding to $H'$ is equal to 1 if and only if $H'$ is a disjoint
union of (non necessarily connected) graphs $H_1',H_2',\ldots,H_k'\in \Ho'$ such that each $H_i'$
appears in a different component of $\mathcal{C}_v$. 
The consistency of the binary arrays $c(v)$ is verified locally as
follows. For each vertex $v$ of $G$, the local verifier at $v$
considers the binary arrays $c(u)$, for all children $u$ of $v$ (and
the binary array $b(v)$, if $v$ is equal to some root $r_C$). For any
$H'\in\Ho'$, the local verifier at $v$ checks whether $H'$ can be
written as a disjoint union of graphs $H_1',H_2',\ldots,H_k'\in \Ho'$ such that each $H_i'$
appears in a different array among the children of $v$ (plus in
$b(v)$, if $v$ is a root of some component $C$). The local verifier at
$v$ then checks whether this is consistent with the entry
corresponding to $H'$ in $c(v)$. Clearly, all the vertices accept if
and only if the information is consistent along the spanning tree, and
it follows that the local verifier at the root $r$ can check for each $H\in
\Ho\subseteq \Ho'$, whether the entry of $c(r)$ corresponding to $H$
is equal to 0 or 1. It follows that the local verifier at $r$ can
check whether $G[X_i]$ is $\Ho$-free (and accept the instance if
and only if this is the case).

\medskip

It follows from
the definition of our scheme that for any $G\in \Po$, the local
verifier of each vertex of
$G$ accepts the instance.

\medskip

Consider now some graph $G$ together with some proof $P'$ such that
the local verifier $\mathcal{A}(G,P',v)$ at each vertex $v$ of $G$ accepts
the instance. For any $1\le i \le s$, let $X_i$ be the set of vertices
$v$ such that $i\in I(v)$ (where $I(v)$ is given by the proof $P'(v)$), and let $F_i=E(G)\setminus
                E(G[X_i])$. As in the proof of Theorem~\ref{thm:lcertsum}, the property
                that every vertex $v \in V(G)$ belongs to at least $(1 -
                  \delta)s$ of sets $X_i$ implies that there is an
                  index $1\le i \le s$ such that $|F_i|\le \eps
                  \e(G)$.

By the properties of the local certificates, each component of $G\setminus F_i=G[X_i]$ has tree-depth at most $d$, and thus $G\setminus F_i$ has
tree-depth at most $d$. Moreover, our local certificates imply that
$G\setminus F_i$ is $\Ho$-free. Since 
$\Po$ is monotone, $G\setminus F_i$ is in $\Po$. It follows that $G$ is $\eps$-close from $\Po$. Taking the
contrapositive, this shows
that if a graph is $\eps$-far  from $\Po$, then at least one local
verifier will reject the instance. This concludes the proof of Theorem~\ref{thm:lcert}.
\hfill $\Box$

\section{Conclusion}\label{sec:cl}

In this paper we proved that for any proper minor-closed class $\Go$,
using constantly many queries to the
random neighbor oracle, a tester can decide with good probability whether
an input graph $G\in \Go$
satisfies some fixed monotone property $\Po$, or is $\eps$-far
from $\Po$. Given the level of generality of the result it is to be expected
that no explicit bounds on the query complexity are given. However, we can
give explicit estimates on the query complexity for specific
properties. For instance, it follows from the bounds of
\cite[Corollary 35]{DS20}, combined with Lemma~\ref{l:unsplit} and our
proof of Theorem~\ref{thm:main}, that 3-colorability can be tested in
planar graphs with $\twr(\text{poly}(1/\eps))$
queries to the random neighbor oracle. This can be extended to testing
$k$-colorability in $K_t$-minor free graphs, for any $k$ and $t$, at
the expense of a significant increase in the height of the tower
function, by combining the results of \cite{DS20} with the main result
of \cite{DDOSRS} (the bounds there are not explicit as a function of
$t$, but can be made explicit using results from the Graph Minor series).
This is to be compared with the main result of \cite{CMOS19}, that
2-colorability can be tested with $2^{2^{\text{poly}(1/\eps)}}$
queries in planar graphs.  It is a natural problem to understand
whether these properties can be tested with $\text{poly}(1/\eps)$
queries to the random neighbor oracle, and more generally to develop
techniques for proving finer lower bounds on the query complexity of
monotone properties in this model (see~\cite{BKS22} for recent results
in this direction in the bounded degree model).

\medskip

\subsubsection*{Acknowledgements} This work was initiated during the Graph
Theory workshop in 
Oberwolfach, Germany, in January 2022. The authors would like to thank the
organizers and participants for all the discussions and nice
atmosphere (and in particular Gwena\"el Joret, Chun-Hung Liu, and Ken-ichi
Kawarabayashi for the discussions related to the topic of this
paper). The authors would also like to thank G\'abor Elek for his
remarks on an earlier version of this manuscript, and his suggestion
to replace minor-closed classes by classes of bounded asymptotic
dimension in Theorem~\ref{thm:lcertsum}.


\end{document}